\documentclass[11pt]{article}

\usepackage{amssymb}
\usepackage{amsmath}
\usepackage{amsthm}
\usepackage[dvips]{graphics,color}
\usepackage{epsfig}
\usepackage{a4wide}
\def\E{\mathbb{E}}
\def\var{\text{Var}}

\def\N{N}
\def\W{W}
\def\IW{IW}
\def\B{B}

\def\T{t}
\def\GIW{GIW}
\def\GW{GW}
\def\U{\mathcal{U}}
\newtheorem{thm}{Theorem}
\newtheorem{lem}{Lemma}

\begin{document}

\title{Multivariate stochastic volatility using state space models}
\author{K. Triantafyllopoulos\footnote{Department of Probability and Statistics,
Hicks Building, University of Sheffield, Sheffield S3 7RH, UK,
Email: {\tt k.triantafyllopoulos@sheffield.ac.uk}, Tel: +44 114 222
3741, Fax: +44 114 222 3759.}\\ Department of Probability and
Statistics, University of Sheffield, Sheffield, UK}

\date{\today}

\maketitle

\begin{abstract}

A Bayesian procedure is developed for multivariate stochastic
volatility, using state space models. An autoregressive model for
the log-returns is employed. We generalize the inverted Wishart
distribution to allow for different correlation structure between
the observation and state innovation vectors and we extend the
convolution between the Wishart and the multivariate singular beta
distribution. A multiplicative model based on the generalized
inverted Wishart and multivariate singular beta distributions is
proposed for the evolution of the volatility and a flexible
sequential volatility updating is employed. The proposed algorithm
for the volatility is fast and computationally cheap and it can be
used for on-line forecasting. The methods are illustrated with an
example consisting of foreign exchange rates data of 8 currencies.
The empirical results suggest that time-varying correlations can be
estimated efficiently, even in situations of high dimensional data.

\textit{Some key words:} volatility, multivariate, GARCH, time
series, state space model, Bayesian forecasting, dynamic linear
model, Kalman filter, generalized Wishart distribution.

\end{abstract}

\section{Introduction}\label{model}

Consider that the $p$-variate time series $\{y_t\}_{t=1,\ldots,N}$
is generated from the multivariate state space model
\begin{equation}\label{model2}
y_t=\theta_t+\Sigma_t^{1/2}\epsilon_t\quad \textrm{and} \quad
\theta_t = \phi\theta_{t-1}+\Omega_t^{1/2}\omega_t,
\end{equation}
where the innovations $\{\epsilon_t\}_{t=1,\ldots,N}$ and
$\{\omega_t\}_{t=1,\ldots,N}$ are individually and mutually
uncorrelated, following the $p$-variate Gaussian distributions
$\epsilon_t\sim \N_p(0,I_p)$ and $\omega_t\sim \N_p(0,I_p)$, for
$\Sigma_t^{1/2}$ being the symmetric square root of $\Sigma_t$
(Gupta and Nagar, 1999) and $I_p$ denotes the $p\times p$ identity
matrix. Typically, at time $t$, $y_t$ will represent the log-returns
of some assets or exchange rates or any other financial time series.
$\Sigma_t$ is the volatility matrix at time $t$ and interest is
placed on its estimation, while $\Omega_t$ is a non-negative
definite matrix. An evolutionary law for $\Sigma_t$ and a density
for the initial state $\theta_0$ have to be defined. It is
worthwhile to note that several volatility models can be obtained
from the formulation of model (\ref{model2}). For example, for
$\phi=1$, $\theta_0=\theta$ (with probability 1), and $\Omega_t=0$,
one obtains the volatility model
$y_t=\theta+\Sigma_t^{1/2}\epsilon_t$. Then, depending on the
evolution law for $\Sigma_t$, one can obtain multivariate GARCH
(MGARCH) type models (Bauwens {\it et al.}, 2006) or multivariate
stochastic volatility (MSV) models (Asai {\it et al.}, 2006;
Maasoumi and McAleer, 2006).

The purpose of this paper is to develop an estimation procedure that
will allow fast and efficient estimation of $\Sigma_t$ and
forecasting of $y_t$. Our motivation stems from work on MSV models
that experience problems due to the simulation-based estimation
procedures they use, see e.g. Uhlig (1997), Aguilar and West (2000),
Philipov and Glickman (2006), and references therein. The fast
estimation procedures, proposed in this paper, aim to achieve high
computational savings (which is necessary in high dimensional data)
and yet enjoy the sophistication of stochastic volatility. For this
to commence, one needs to define $\Omega_t$ and to propose an
evolutionary law for $\Sigma_t$. Since $\theta_t$ is unobserved
signal, $\Omega_t$ is suggested to be specified, rather than
estimated from the data, as the latter will require to resort to
simulation-based estimation techniques (e.g. MCMC or EM algorithm),
and this can cause significant delays of the estimation of the
volatility. In this paper we adopt Bayesian estimation, for which we
aim to specify a prior distribution for $\Sigma_0$. Then it is
desirable, in order to develop conjugate analyses that will
facilitate fast estimation, to define $\Omega_t$ to be proportional
to $\Sigma_t$. Indeed, in the context of a time-invariant volatility
$\Sigma_t=\Sigma$, this setting is well known (Harvey, 1989; Durbin
and Koopman, 2001; L\"utkepohl, 2007). In this setting
$\Omega_t=\Omega \Sigma_t$, where $\Omega>0$ is known. However, this
is overly restrictive, as the correlation matrix of
$\Sigma_t^{1/2}\epsilon_t$ is the same as the correlation matrix of
$\Omega_t^{1/2}\omega_t$. In this paper we define
$\Omega_t=\Sigma_t^{1/2}\Omega\Sigma_t^{1/2}$, where now $\Omega$ is
a known non-negative definite matrix (later in the paper we explain
how this matrix can be specified). This setting clearly encompasses
the situation when $\Omega$ is scalar, and it allows more general
and more flexible estimation.

For the volatility covariance matrix $\Sigma_t$, we propose a
multiplicative stochastic law of its precision $\Sigma_t^{-1}$, i.e.
\begin{equation}\label{evol}
\Sigma_{t}^{-1}=k\U(\Sigma_{t-1}^{-1})'B_t\U(\Sigma_{t-1}^{-1}),\quad
t=1,\ldots,N,
\end{equation}
where $k=\{\delta(1-p)+p\}\{\delta(2-p)+p-1\}^{-1}$, for a discount
factor $0<\delta<1$, and $\U(\Sigma_{t-1}^{-1})$ denotes the unique
upper triangular matrix based on the Choleski decomposition of
$\Sigma_{t-1}^{-1}$. Here $B_t$ is a $p\times p$ random matrix
following the multivariate singular beta distribution
$B_t\sim\B(m/2,1/2)$, where $m=\delta(1-\delta)^{-1}+p-1$. Some
details of this distribution can be found in the appendix (see Lemma
\ref{lemma}), but for more details the reader is referred to Uhlig
(1994), D\'{i}az-Garc\'{i}a and Guti\'{e}rrez (1997), and Srivastava
(2003). In Section \ref{s2s2} it is shown that when $\Omega=I_p$ and
with $\E(.)$ denoting expectation, we have
$\E(\Sigma_{t}^{-1}|y^{t-1})=\E(\Sigma_{t-1}^{-1}|y^{t-1})$, while
the respective covariance matrix at $t$ is increased of that at time
$t-1$; in this case $\Sigma_t^{-1}$ follows a random walk. When
$\Omega$ is a covariance matrix, then the evolution (\ref{evol})
suggests approximately a random-walk type process for
$\Sigma_t^{-1}$. The choice of $k$ is made in order to accommodate
the above random walk equations. It should be noted that, if $p=1$,
it is $k=1/\delta$, and (\ref{evol}) is reduced to
$\Sigma_t=\delta\Sigma_{t-1} B_t^{-1}$ (Uhlig, 1994;
Triantafyllopoulos, 2007). In order to accommodate for the
definition of $\Omega_t$, we generalize the Wishart distribution and
we extend its convolution with the multivariate singular beta, which
was first proved in Uhlig (1994). In order to support conjugate
analysis, this generalization is necessary, because of the
definition of $\Omega_t$. Finally, at time $t=0$, $\theta_0$ is
assumed to be uncorrelated with $\{\epsilon_t\}_{t=1,2,\ldots,N}$
and $\{\omega_t\}_{t=1,\ldots,N}$ and it is assumed that
$\theta_0\sim\N_p(m_0,\Sigma_0^{1/2}P_0\Sigma_0^{1/2})$, for some
known prior mean vector $m_0$ and covariance matrix $P_0$. The
scalar constant $\phi$ is assumed known. Compared with existing
MGARCH and MSV models, a major advantage of the proposed
methodology, is that the likelihood function is provided in closed
form. This can facilitate model comparison, but it can also be used
as a means for the choice of the parameters, without the need to
resort to numerical methods in order to maximize the likelihood
function (more details on this are provided via the data analysis in
Section \ref{examples}).

The remaining of the paper is organized as follows. The following
section generalizes the inverted Wishart distribution and discusses
some properties of the new distribution. In Section \ref{s2s2} the
main algorithm of the volatility is developed and Section
\ref{examples} analyzes the volatility of foreign exchange rates
data. The findings of the paper are summarized in Section
\ref{discussion} and the appendix includes all proofs of arguments
in Sections \ref{s2s1} and \ref{s2s2}.

\section{Generalized inverted Wishart distribution}\label{s2s1}

Let $X\sim \IW_p(n,A)$ denote that the matrix $X$ follows an
inverted Wishart distribution with $n$ degrees of freedom and with
parameter matrix $A$. Given $A$, we use the notation $|A|$ for the
determinant of $A$ and the notation $\textrm{etr}(A)$ for the
exponent of the trace of $A$. The following theorem introduces a new
distribution generalizing the inverted Wishart distribution.

\begin{thm}\label{lem1}
Consider the $p\times p$ random covariance matrix $X$ and denote
with $X^{1/2}$ the symmetric square root of $X$. Given $p\times p$
covariance matrices $A$ and $S$ and a positive scalar $n>2p$, define
$Y=X^{1/2}A^{-1}X^{1/2}$ so that $Y$ follows an inverted Wishart
distribution $Y\sim IW_p(n,S)$. Then the density function of $X$ is
given by
$$
p(X)=\frac{|A|^{(n-p-1)/2}|S|^{(n-p-1)/2}}{2^{p(n-p-1)/2}\Gamma_p\{(n-p-1)/2\}
|X|^{n/2}}\textrm{etr}(-AX^{-1/2}SX^{-1/2}/2),
$$
where $\Gamma_p(.)$ denotes the multivariate gamma function.
\end{thm}

The distribution of the above theorem proposes a generalization of
the inverted Wishart distribution, since if $A=I_p$ we have $X\sim
\IW_p(n,S)$ and if $S=I_p$, we have $X\sim \IW_p(n,A)$. This is
clearly a different generalization of other generalizations of the
inverted Wishart distribution, see Dawid and Lauritzen (1993), Brown
{\it et al.} (1994), Roverato (2002), and Carvalho and West (2007).
In the following we refer to the distribution of Theorem \ref{lem1}
as {\it generalized inverted Wishart} distribution, and we write
$X\sim \GIW_p(n,A,S)$. The next result gives some expectations of
the $\GIW$ distribution.
\begin{thm}\label{lem2a}
Let $X\sim \GIW_p(n,A,S)$ for some known $n,A$ and $S$. Then we have
\begin{enumerate}
\item [(a)] $\E(X^{1/2}S^{-1}X^{1/2})=(n-2p-2)^{-1}A$;
$\E(X^{-1/2}SX^{-1/2})=(n-p-1)A^{-1}$; \item [(b)]
$\E|X|^{\ell}=2^{-p\ell}[\Gamma_p
\{(n-p-1)/2\}]^{-1}\Gamma_p\{(n-2\ell-p-1)/2\}|A|^{\ell}|S|^{\ell}$,
\end{enumerate}
where $\E(.)$ denotes expectation and $0<\ell<(n-2p)/2$.
\end{thm}
The following property reflects on the symmetry of $A$ and $S$ in
the $\GIW$ distribution.
\begin{thm}\label{th:commute}
If $X\sim \GIW_p(n,A,S)$, for some known $n$, $A$ and $S$, then
$X\sim \GIW_p(n,S,A)$.
\end{thm}

We motivate the estimator $\widehat{X}(A,S)$ of $X\sim
\GIW_p(n,A,S)$ as follows. The estimator should be a symmetric
positive definite matrix and for $A$ and $S$ being matrices, one
possibility is $\widehat{X}(A,S)=kA^{1/2}SA^{1/2}$, for a known
constant $k$. This estimator equals the expectation of the inverted
Wishart distribution $\IW_p(k+2p+2,A^{1/2}SA^{1/2})$. Since in
general $A^{1/2}SA^{1/2}\neq S^{1/2}AS^{1/2}$, a similar estimator
for $X$ can be considered as
$\widehat{X}^*(A,S)=k^*S^{1/2}AS^{1/2}$, for some constant $k^*$. We
propose that the desired estimator for $X$ should satisfy the
following requirements:

\begin{description}
\item (1) In the univariate case $(p=1)$ the estimator should be
$\widehat{X}(A,S)=AS/(n-4)$; \item (2) The estimator should be
symmetric in $A$ and $S$, i.e. $\widehat{X}(A,S)=\widehat{X}(S,A)$;
\item (3) If $A=I_p$ the estimator should reduce to the expectation
from the inverted Wishart density $\IW_p(n,S)$, i.e.
$\widehat{X}(A,S)=(n-2p-2)^{-1}S$; If $S=I_p$ the estimator should
reduce to the expectation from the inverted Wishart density
$\IW_p(n,A)$, i.e. $\widehat{X}(A,S)=(n-2p-2)^{-1}A$.
\end{description}

Now we propose the estimator
\begin{equation}\label{sp:estim}
\widehat{X}(A,S)=\frac{1}{2n-4p-4}\left(S^{1/2}AS^{1/2}+A^{1/2}SA^{1/2}\right),
\end{equation}
for which we can see that (1)-(3) are satisfied.

It is also easy to verify that if $X\sim\GIW_p(n,A,S)$, then the
density of $Y=X^{-1}$ is
$$
p(Y) = \frac{|A|^{(n-p-1)/2} |S|^{(n-p-1)/2}
|Y|^{(n-2p-2)/2}}{2^{p(n-p-1)/2}\Gamma_p\{(n-p-1)/2\} }
\textrm{etr}(-AY^{1/2}SY^{1/2}/2).
$$
This distribution generalizes the Wishart distribution; we will say
that $Y$ follows the {\it generalized Wishart} distribution with
$n-p-1$ degrees of freedom, covariance matrices $A^{-1}$ and
$S^{-1}$, and we will write $Y\sim\GW_p(n-p-1,A^{-1},S^{-1})$. It is
easy to see that when $A=I_p$ or $S=I_p$, the above density reduces
to a Wishart density. Again our terminology and notation, should not
cause any confusion with other generalizations of the Wishart
distribution, proposed in the literature (Letac and Massam, 2004).

The next theorem is a generalization of the convolution of the
Wishart and multivariate singular beta distributions (Uhlig, 1994).
For some integers $m,n$, denote with $\B_p(m/2,n/2)$ the
multivariate singular beta distribution with $m$ and $n$ degrees of
freedom. The density of this distribution is given in the appendix
(see Lemma \ref{lemma}) and more details can be found in Uhlig
(1994), D\'{i}az-Garc\'{i}a and Guti\'{e}rrez (1997), and Srivastava
(2003).

\begin{thm}\label{th:uhlig}
Let $p$ and $n$ be positive integers and let $m>p-1$. Let
$H\sim\GW_p(m+n,A,S)$ and $B\sim\B_p(m/2,n/2)$ be independent, where
$A$ and $S$ are known covariance matrices. Then
$$
G\equiv \mathcal{U}(H)'B\mathcal{U}(H) \sim \GW_p(m,A,S),
$$
where $\mathcal{U}(H)$ denotes the upper triangular matrix of the
Choleski decomposition of $H$.
\end{thm}

\section{Estimation}\label{s2s2}

\subsection{The main algorithm}

In this section we consider estimation for model (\ref{model2}),
where $\Sigma_t$ follows the evolution (\ref{evol}). The prior
distributions of $\theta_0|\Sigma_0$ and $\Sigma_0$ are chosen to be
Gaussian and a generalized inverted Wishart respectively, i.e.
\begin{gather}
\theta_0|\Sigma_0\sim
\N_p(m_0,\Sigma_0^{1/2}P_0\Sigma_0^{1/2})\quad\textrm{and}\quad
\Sigma_0\sim \GIW_p(n_0,Q^{-1},S_0),\label{eq8}
\end{gather}
for some known parameters $m_0$, $P_0=p_0I_p$, $n_0>2p+2$ and $S_0$.
$Q$ is the limit of $Q_{t-1}(1)=P_{t-1}+\Omega+I_p$, where $P_t$ is
a known covariance matrix. The next result shows that the limit of
$P_t$ (and hence the limit of $Q_{t-1}(1)$) exist and it provides
the value of this limit as a function of $\phi$ and $\Omega$.

\begin{thm}\label{lem:limit}
If $P_{t}=R_{t}(R_{t}+I_p)^{-1}$, with $R_{t}=\phi^2P_{t-1}+\Omega$,
where $\Omega$ is a positive definite matrix and considering the
prior $P_0=p_0I_p$, for a known constant $p_0>0$, it is
$$
P=\lim_{t\rightarrow\infty}P_t=\frac{1}{2\phi^2}\left[
\left\{(\Omega+(1-\phi^2)I_p)^2+4\Omega\right\}^{1/2}-\Omega-(1-\phi^2)I_p\right],
$$
for $\phi\neq 0$ and $P=\Omega (\Omega+I_p)^{-1}$, for $\phi=0$.
\end{thm}
This result generalizes relevant limit results for the univariate
random walk plus noise model (Anderson and Moore, 1979, page 77;
Harvey, 1989, page 119).

Let $Y\sim \T_p(n,m,P)$ denote that the $p$-dimensional random
vector $Y$ follows a multivariate Student $t$ distribution with $n$
degrees of freedom, mean $m$ and scale or spread matrix $P$ (Gupta
and Nagar, 1999, Chapter 4). The next result gives an approximate
Bayesian algorithm for the posterior distributions of $\theta_{t}$
and $\Sigma_{t}$ as well as for the one-step forecast distribution
of $y_{t}$.

\begin{thm}\label{th3}
In the multivariate state space model (\ref{model2}) with evolution
(\ref{evol}), let the initial priors for $\theta_0|\Sigma_0$ and
$\Sigma_0$ be specified as in equation (\ref{eq8}). The one-step
forecast and posterior distributions are approximately given, for
each $1\leq t\leq N$, as follows:
\begin{enumerate}
\item [(a)] One-step forecast at time $t$: $\Sigma_{t}|y^{t-1}\sim\GIW_p(\delta
(1-\delta)^{-1}+2p,Q^{-1},k^{-1} S_{t-1})$ and $y_{t}|y^{t-1}\sim
\T_p(\delta(1-\delta)^{-1},m_{t-1},k^{-1}S_{t-1})$, where
$k=(\delta(1-p)+p)(\delta(2-p)+p-1)^{-1}$ and $\delta$, $S_{t-1}$,
$m_{t-1}$ are known at time $t-1$. \item [(b)] Posteriors at $t$:
$\theta_{t}|\Sigma_{t},y^{t}\sim
\N_p(m_{t},\Sigma_{t}^{1/2}P_{t}\Sigma_{t}^{1/2})$\\ and
$\Sigma_{t}|y^{t}\sim \GIW((1-\delta)^{-1}+2p,Q^{-1},S_{t})$, with
\begin{gather*}
m_{t}=m_{t-1}+A_{t}e_{t},\quad
P_{t}=(\phi^2P_{t-1}+\Omega)(\phi^2P_{t-1}+\Omega+I_p)^{-1},\\
e_{t}=y_{t}-m_{t-1}, \quad S_{t}=k^{-1}S_{t-1}+e_{t}e_{t}',
\end{gather*}
where $A_{t}=\Sigma_{t}^{1/2}P_{t}\Sigma_{t}^{-1/2}$ is approximated
by $A_{t}^{*}=(S_{t}^*)^{1/2}P_{t}(S_{t}^*)^{-1/2}$, with
$S_{t}^*=\widehat{\Sigma}(Q^{-1},S_{t})$ the estimator of
$\Sigma_{t}|y^{t}$ (see equation (\ref{sp:estim})) and
$Q_{t-1}(1)=P_{t-1}+\Omega+I_p$ being approximated by its limit
$Q=P+\Omega+I_p$, where $P$ is given by Theorem \ref{lem:limit}.
\end{enumerate}
\end{thm}

From Theorem \ref{th3} we have that the one-step forecast vector
mean and covariance matrix of $y_{t}$ are
$$
y_{t-1}(1)=\E(y_{t}|y^{t-1})=m_{t-1}\quad\textrm{and}\quad
\var(y_{t}|y^{t-1})=
\frac{k^{-1}S_{t-1}}{\delta(1-\delta)^{-1}-2}=\frac{(1-\delta)S_{t-1}}{(3\delta-2)k},
$$
for $\delta>2/3$.

We note that $k>1$, since $\delta(1-p)+p>\delta(2-p)+p-1$, for any
$0<\delta<1$ and so, if we expand $S_t$ as
$$
S_t=k^{-t}S_0+\sum_{i=1}^tk^{i-t}e_ie_i',
$$
for large $t$, we can approximate $S_t$ by
$\sum_{i=1}^tk^{i-t}e_ie_i'$. The observation that $k^{-1}<1$ is
important, because otherwise $S_t$ could tend to infinity.

From Theorem \ref{th3}, if $\Omega=I_p$, then
$\Sigma_{t}|y^{t-1}\sim\IW_p(\delta(1-\delta)^{-1}+2p,k^{-1}Q^{-1}S_{t-1})$
and
$\Sigma_{t-1}|y^{t-1}\sim\IW_p((1-\delta)^{-1}+2p,Q^{-1}S_{t-1})$,
where now $Q$ is a variance. Thus
$\Sigma_{t}^{-1}|y^{t-1}\sim\W_p(\delta(1-\delta)^{-1}+p-1,kQS_{t-1}^{-1})$
and
$\Sigma_{t-1}^{-1}|y^{t-1}\sim\W_p((1-\delta)^{-1}+p-1,QS_{t-1}^{-1})$
so that
\begin{equation}\label{rwalk1}
\E(\Sigma_{t}^{-1}|y^{t-1})=\left(\frac{\delta}{1-\delta}+p-1\right)
k Q S_{t-1}^{-1} = \left(\frac{1}{1-\delta}+p-1\right) Q
S_{t-1}^{-1}=\E(\Sigma_{t-1}^{-1}|y^{t-1}),
\end{equation}
with $k=(\delta(1-p)+p)(\delta(2-p)+p-1)^{-1}$. From the Wishart
densities it follows that, given $y^{t-1}$,
\begin{equation}\label{rwalk2}
\var\{\textrm{vecp}(\Sigma_{t}^{-1})|y^{t-1}\}\geq
\var\{\textrm{vecp}(\Sigma_t^{-1})|y^{t-1}\},
\end{equation}
in the sense that
$\var\{\textrm{vecp}(\Sigma_{t}^{-1})|y^{t-1}\}-\var\{\textrm{vecp}(\Sigma_t^{-1})|y^{t-1}\}$
is a non-negative definite matrix, where
$\textrm{vecp}(\Sigma_t^{-1})$ denotes the column stacking operator
of $\Sigma_t^{-1}$. Equations (\ref{rwalk1}) and (\ref{rwalk2}) show
that when $\Omega=I_p$, $\Sigma_{t}^{-1}$ follows a random walk type
evolution. When $\Omega$ is a covariance matrix we can see that
$$
\E(\Sigma_{t}^{-1/2}Q\Sigma_{t}^{-1/2}|y^{t-1})=\E(\Sigma_{t-1}^{-1/2}Q\Sigma_{t-1}^{-1/2}
|y^{t-1})
$$
and
$\var\{\textrm{vecp}(\Sigma_{t}^{-1/2}Q\Sigma_{t}^{-1/2}|y^{t-1})\}\geq
\var\{\textrm{vecp}(\Sigma_{t-1}^{-1/2}Q\Sigma_{t-1}^{-1/2}
|y^{t-1}\}$. The proof is by noting that
$\Sigma_{t}^{-1/2}Q\Sigma_{t}^{-1/2}|y^{t-1}\sim\W_p(\delta(1-\delta)^{-1}+p-1,kS_{t-1}^{-1})$
and
$\Sigma_{t-1}^{-1/2}Q\Sigma_{t-1}^{-1/2}|y^{t-1}\sim\W_p((1-\delta)^{-1}+p-1,S_{t-1}^{-1})$
and following a similar argument as in equations (\ref{rwalk1}) and
(\ref{rwalk2}). Hence when $\Omega$ is a covariance matrix
$\Sigma_{t}^{-1/2}Q\Sigma_{t}^{-1/2}$ follows a random walk type
evolution and this motivates the adoption of evolution (\ref{evol}).
Equation (\ref{rwalk1}) shows that under the definition of $k$, the
expectation of $\Sigma_{t}^{-1}$ equals the expectation of
$\Sigma_{t-1}^{-1}$, while the respective variances are increased
from time $t-1$ to $t$.

\subsection{Performance measures}

In this section we discuss several performance measures for model
(\ref{model2}). We start giving the log-likelihood function of
$\Sigma_t$.

\begin{thm}\label{th2}
In the state space model (\ref{model2}) with evolution (\ref{evol})
denote with $\ell(\Sigma_1,\ldots,\Sigma_N;y^N)$ the log-likelihood
function of $\Sigma_1,\ldots,\Sigma_N$, based on data
$y^N=\{y_1,\ldots,y_N\}$. Then it is
\begin{eqnarray*}
\ell(\Sigma_1,\ldots,\Sigma_N;y^N)  &=& c -
\frac{1}{2}\sum_{t=1}^Ne_t'\Sigma_t^{-1/2}Q\Sigma_t^{-1/2}e_t -
\frac{2\delta-1}{1-\delta}\sum_{t=1}^N\log |\U(\Sigma_{t-1}^{-1})| \\
&& - \frac{p}{2} \sum_{t=1}^N\log |L_t|
-\frac{3\delta-2}{2(1-\delta)}\sum_{t=1}^N\log |\Sigma_t|,
\end{eqnarray*}
and
\begin{eqnarray*}
c&=& -Np\log\pi -\frac{N}{2}\log |Q| - \frac{Np}{2}\log k \\ &&
+N\log\left\{ \Gamma_p\left(\frac{\delta(1-p)+p}{2(1-\delta)}\right)
\bigg/ \Gamma_p\left(\frac{\delta(2-p)+p-1}{2(1-\delta)}\right)
\right\},
\end{eqnarray*}
where $\delta>2/3$, $k=\{\delta(1-p)+p\}\{\delta(2-p)+p-1\}^{-1}$
and $L_t$ is the diagonal matrix with diagonal elements the positive
eigenvalues of
$I_p-k^{-1}\{\U(\Sigma_{t-1}^{-1})'\}^{-1}\Sigma_t^{-1}\{\U(\Sigma_{t-1}^{-1})\}^{-1}$,
with $\Sigma_t^{-1}=\U(\Sigma_t^{-1})'\U(\Sigma_t^{-1})$.
\end{thm}

The choice of $\delta$, $\Omega$ and the priors $m_0$, $p_0$, and
$S_0$ can be done by either maximizing the log-likelihood function
or optimizing performance measures, such as the mean of square
one-step forecast errors (MSE), the mean of square standardized
one-step forecast errors (MSSE), the mean absolute deviation (MAD),
and the mean one-step forecast error (ME). The priors can be set
using historical data, but a general guideline suggests $m_0=0$,
$p_0=1000$ and $S_0=I_p$. In any case these initial settings are not
critical to the performance of the model, especially given plethora
of data. It then remains to specify $\Omega$ and $\delta$. Given
data $y_1,\ldots,y_N$, the definition of the above mentioned
performance measures are
\begin{gather*}
\textrm{MSE}=\frac{1}{N}\sum_{t=1}^N\left(e_{1t}^2,\ldots,e_{pt}^2\right)',
\quad \textrm{MSSE} = \frac{1}{N} \sum_{t=1}^N \left(u_{1t}^2,
\ldots , u_{pt}^2 \right)',
\\ \textrm{MAD} = \frac{1}{N} \sum_{t=1}^N
\left(\textrm{mod} (e_{1t}) , \ldots, \textrm{mod} (e_{pt})\right)',
\quad \textrm{ME} = \frac{1}{N} \sum_{t=1}^N \left(e_{1t}, \ldots,
e_{pt} \right)',
\end{gather*}
where $e_t=(e_{1t},\ldots,e_{pt})'$ is the one-step forecast error
vector, $\textrm{mod}(e_{jt})$ denotes the modulus of $e_{jt}$
$(j=1,\ldots,p)$ and $u_t=(u_{1t},\ldots,u_{pt})'$ is the
standardized one-step forecast error vector, defined by
$$
u_t=\left\{\frac{(1-\delta)S_t}{(3\delta-2)k}\right\}^{-1/2}e_t,
$$
so that $\E(u_t|y^{t-1})=0$ and $\E(u_tu_t'|y^{t-1})=I_p$. Thus, if
the model is a good fit, it should return
$\textrm{MSSE}\approx(1,\ldots,1)'$, $\textrm{ME}\approx (0,\ldots,
0)'$, while $\textrm{MAD}$ and $\textrm{MSE}$ should be as small as
possible.

The MSSE is usually preferred to MSE, because it takes into account
the forecast covariance matrix of the log-returns. However, since
the MSE can be used for comparison of two or more models it is
mentioned here. When we look at the performance of a single model
the MSSE has the ability to judge the goodness of fit in an
effective way. The MAD has a similar performance as the MSE, while
the ME is useful if we wish to check how biased is the estimation
method (Fildes, 1992).

In order to choose the optimal $\Omega$ we propose the following
search procedure. Since $\Omega$ has $p(p+1)/2$ distinct elements,
for relatively large $p$ there are many elements in $\Omega$ to be
optimized. One can reduce the dimensionality of this optimization by
considering a diagonal choice for $\Omega$, writing
$\Omega=\textrm{diag}(w_1,\ldots,w_p)$. Since $0<w_i<\infty$, still
a search procedure for the optimal $w_i$ can be time-consuming. By
defining $Z=\Omega(\Omega+I_p)^{-1}$, we have that $Z$ is also
diagonal and it is $0<Z<I_p$. This means that we can use a grid
search procedure to find the optimal value for $Z$ and then choose
$\Omega=(I_p-Z)^{-1}Z$. For $Z=(z_1,\ldots,z_p)'$ we can use
$z_i=1/10^q,\ldots,(10^q-1)/10^q$, for $i=1,\ldots,p$ and $q$ a
positive integer; for most applications $q=2$ or $q=3$ will suffice.
Then we can readily see that $w_i=z_i/(1-z_i)$, for $i=1,\ldots,p$.
We use this search procedure in the example of Section
\ref{examples}.

\section{FX data analysis}\label{examples}

\begin{figure}[t]
 \epsfig{file=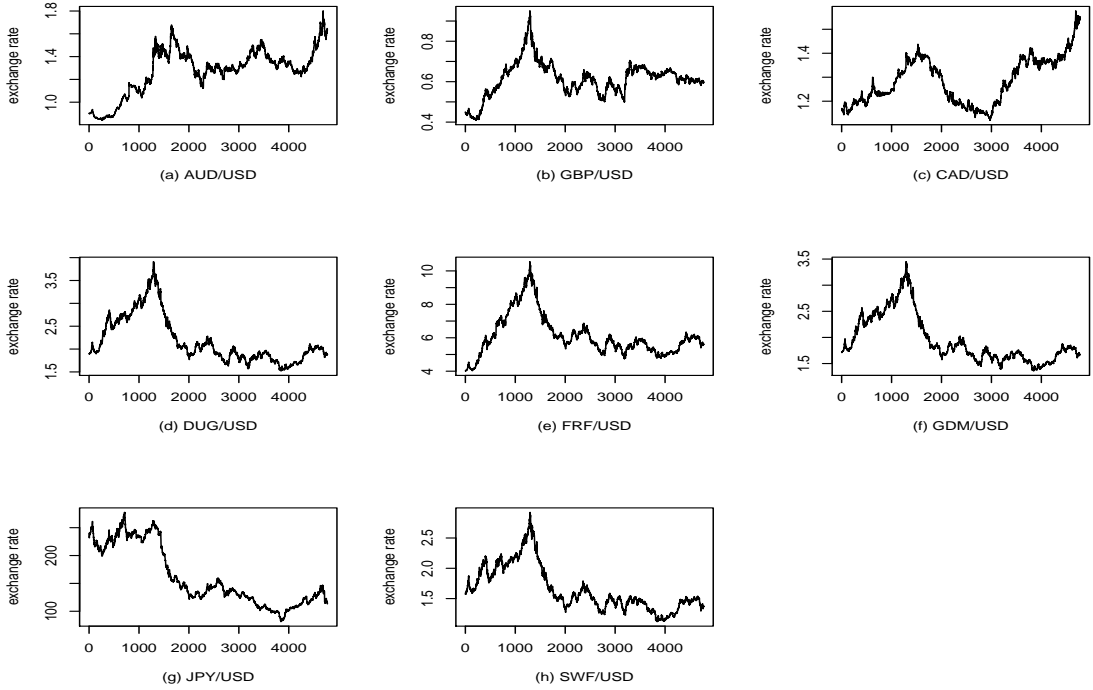, height=10cm, width=15cm}
 \caption{FX data; shown are (a) the AUD/USD exchange rate,
 (b) the GBP/USD rate (c) the CAD/USD rate, (d) the DUG/USD rate,
 (e) the FRF/USD rate, (f) the GDM/USD rate, (g) the JPY/USD rate and
 (h) the SWF/USD rate. }\label{fig3}
\end{figure}

\begin{figure}[t]
 \epsfig{file=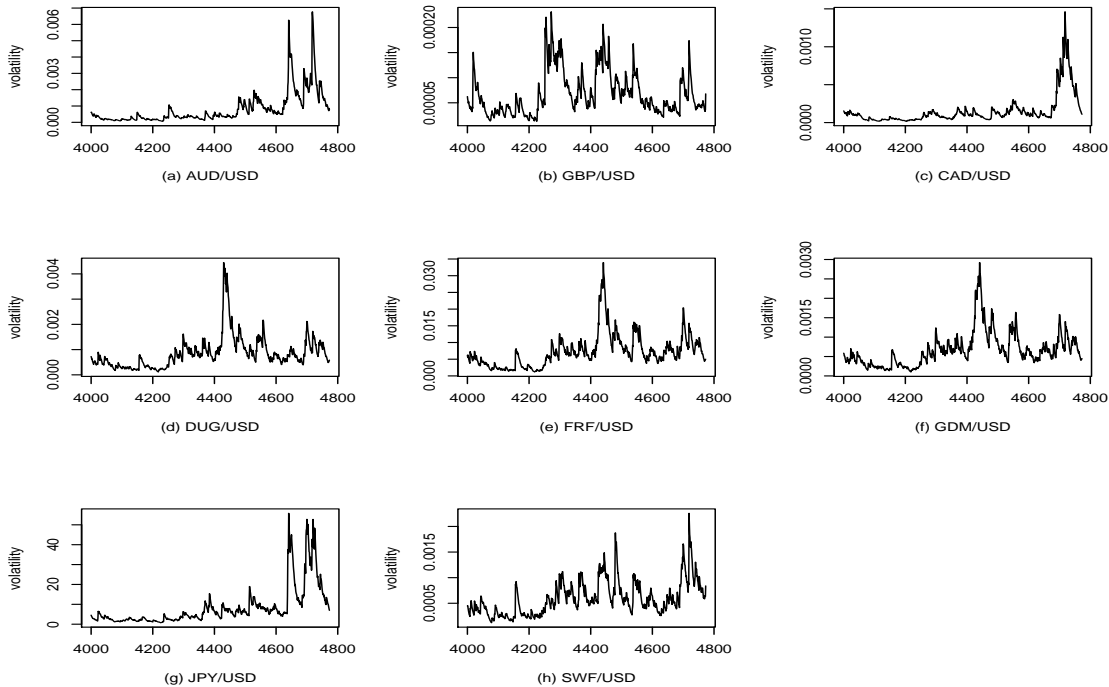, height=10cm, width=15cm}
 \caption{Estimate of the posterior volatility for the FX data. Shown are: (a) the AUD/USD volatility,
 (b) the GBP/USD volatility, (c) the CAD/USD volatility, (d) the DUG/USD volatility,
 (e) the FRF/USD volatility, (f) the GDM/USD volatility, (g) the JPY/USD volatility and
 (h) the SWF/USD volatility.}\label{fig4}
\end{figure}

\begin{figure}[t]
 \epsfig{file=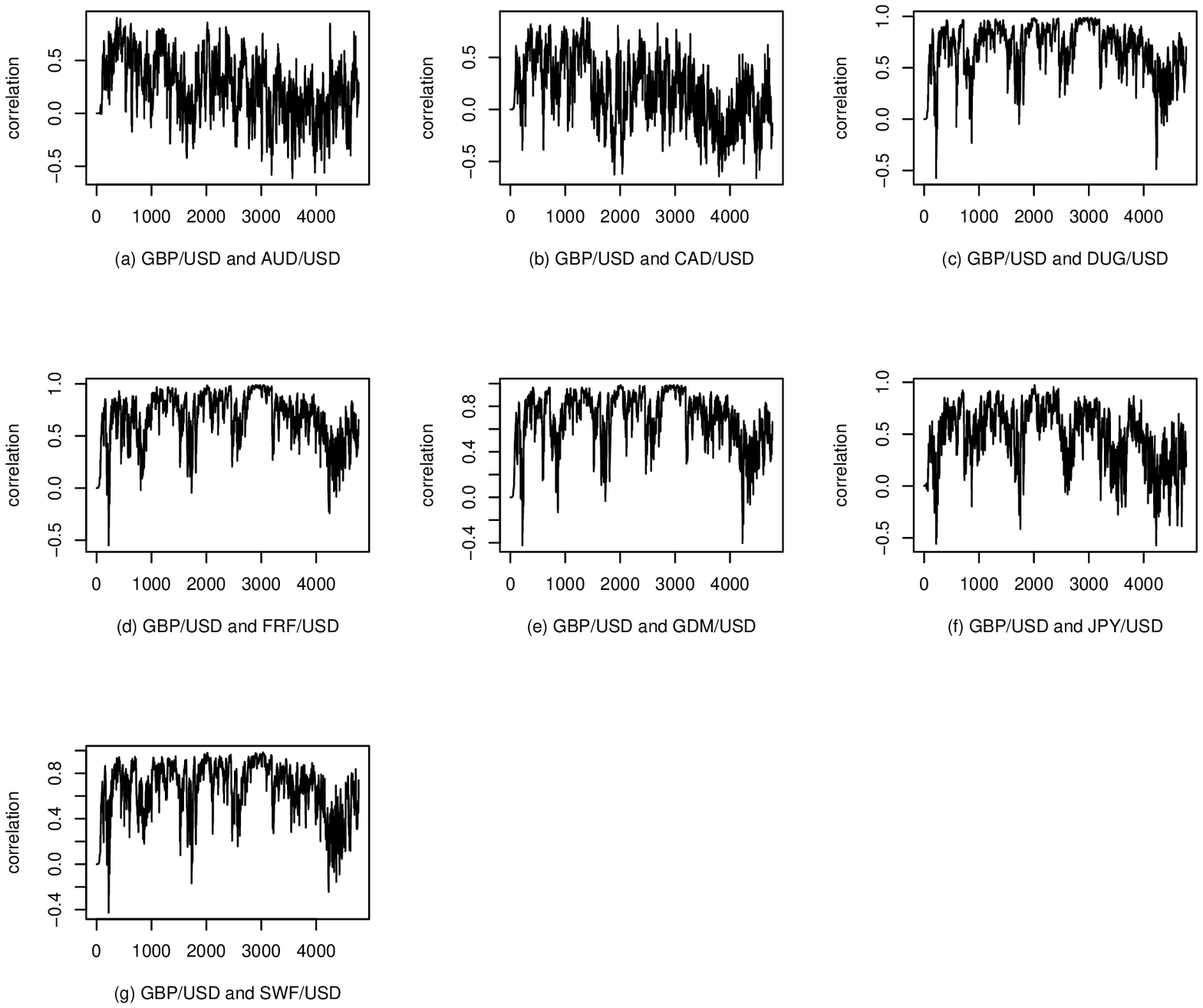, height=10cm, width=15cm}
 \caption{Estimate of the posterior correlation coefficient for the FX data. Shown are: (a) the posterior correlation
 estimate of GBP/USD and AUD/USD, (b) the correlation estimate of GBP/USD and CAD/USD, (c) the correlation of
 GBP/USD and DUG/USD, (d) the correlation of GBP/USD and FRF/USD, (e) the correlation of
 GBP/USD and GDM/USD, (f) the correlation of GBP/USD and JPY/USD, and (g) the correlation of
 GBP/USD and SWF/USD.}\label{fig5}
\end{figure}

In this section we consider foreign exchange rates data (FX) of 8
currencies, namely Australian Dollar vs US Dollar (AUD/USD), British
Pound vs USD (GBP/USD), Canadian Dollar vs USD (CAD/USD), Dutch
Guilder vs USD (DUG/USD), French Franc vs USD (FRF/USD), German
DeutschMark vs USD (GDM/USD), Japan Yen vs USD (JPY/USD), and Swiss
Franc vs USD (SWF/USD). The data are sampled in daily frequency,
from January 1980 to December 1997 and these data are reported in
Franses and van Dijk (2000). We form the log-returns vector series
$y_t=(y_{1t},\ldots,y_{8t})'$, where $y_{1t}$ is the log-returns of
AUD/USD, $\ldots$, $y_{8t}$ is the log-returns of SWF/USD; the data
are plotted in Figure \ref{fig3}. To specify $\Omega$ we used the
log-likelihood criterion with the search procedure of Section
\ref{s2s2}. Using $q=2$, an optimal diagonal matrix $Z$ was
$Z=\textrm{diag}(0.44,0.54,0.56,0.87,0.92,0.52,0.99,0.77)$ and so
the diagonal $\Omega$ that maximizes the log-likelihood function is
$\Omega=\textrm{diag}(0.786,1.174,1.272,6.692,11.500,1.083,99.000,3.348)$,
for $\delta=0.7$, $\phi=1$, $m_0=0$ and $S_0=I_8$. This setting for
$\Omega$ reveals a clear benefit as opposed to a setting
$\Omega=wI_8$, for a known $w\geq 0$, as we can see that the
correlation matrix of $\epsilon_t$ and the correlation matrix of
$\omega_t$ are not the same. Indeed, at the posterior estimate
$S_{4773}^*$ of $\Sigma_{4773}$, we can see that the correlation
matrices of $\epsilon_t$ and $\omega_t$ differ significantly, with
the latter having larger correlations.

For this data set we observed that larger values of $\delta$ (in
particular values of $\delta$ in the range $0.9\leq \delta\leq
0.99$) can not capture the volatility shocks, returning large values
for the MSSE. The log-likelihood function, evaluated at the
posterior estimate of $\Sigma_t$, for the above optimal settings was
$-98438.78$. The four performance measures are
\begin{gather*}
\textrm{MSE}=(0.00009, 0.00002, 0.00001, 0.00031, 0.00244, 0.00026,
1.67707, 0.00022)', \\ \textrm{MSSE}=(0.933, 0.911, 0.980, 0.979,
0.948, 0.940, 0.924, 0.916)', \\ \textrm{MAD}=(0.006, 0.003, 0.003,
0.013, 0.036, 0.012, 0.915, 0.011)', \\ \textrm{ME}=(-3.14\times
10^{-7}, 7.53\times 10^{-7}, -2.73\times 10^{-6}, -2.32\times
10^{-6}, -5.77\times 10^{-6}, -1.82\times 10^{-6}, \\ -4.67\times
10^{-4}, -2.08\times 10^{-8})'
\end{gather*}
The MSSE is slightly under $(1,\ldots,1)'$, which means that the
volatilities are slightly over-estimated. However, looking at the
MSE, MAD, the ME and the log-likelihood function, we consider this
model as acceptable.

For $\Sigma_t=(\sigma_{ij,t})_{i,j=1,\ldots,8}$, Figure \ref{fig4}
shows the posterior volatilities of each of the $\sigma_{ii,t}$
$(i=1,\ldots,8)$, for the last 774 observations, i.e. from $t=4000$
until $t=4773$. Most of the volatilities are small, except for the
JPY/USD, but even for small volatilities Figure \ref{fig4} indicates
clearly the highly volatile periods for each exchange rate. Figure
\ref{fig5} shows the posterior correlations of GBP/USD with the
other rates. This figure confirms that the correlations are
time-varying. By inspecting Figure \ref{fig5} we observe that
GBP/USD is most correlated with DUG/USD, FRF/USD and JPY/USD, while
GBP/USD is least correlated (but still significantly correlated)
with AUD/USD and CAD/USD.

We note that for this data set, there were 4773 time points and for
the 8-dimensional time series $\{y_t\}$, the estimation algorithm
(see Theorems \ref{th3} and \ref{th2}), implemented in {\tt R},
(including the search procedure to maximize the log-likelihood
function) took about 8 minutes to run on a Pentium PC.

\section{Discussion}\label{discussion}

In this paper we have provided a Bayesian analysis for multivariate
stochastic volatility. We propose a generalization of the Wishart
and inverted Wishart distributions and we extend the convolution
between the Wishart and the multivariate singular beta
distributions. This generalization is motivated from the
multivariate random walk plus noise model, which innovation vectors
are desired to have different correlations. The proposed estimation
methodology is delivered in closed form and it is fast and easily
implementable, even for high dimensional data. The log-likelihood of
the volatility is obtained in closed form and this is an important
step forward on multivariate volatility estimation, quoting ``The
estimation of the canonical SV model and its various extensions was
at one time considered difficult since the likelihood function of
these models is not easily calculable.'' from Chib {\it et al.}
(2007). The availability of the log-likelihood function in closed
form allows more efficient model comparisons, e.g. via sequential
likelihood tests or via sequential Bayes' factors (Salvador and
Gargallo, 2004; Triantafyllopoulos, 2006). Moreover, the proposed
model develops a fast Bayesian algorithm not depending on
simulation-based estimation procedures and not requiring many
parameters to be estimated. In the special case where the volatility
of state vector is proportional to the volatility of the observation
vector, the analysis is exact and the inverse of the volatility
matrix follows a Wishart process.

The procedure proposed in this paper attempts to combine the
simplicity of non-iterative algorithms with the sophistication of
stochastic volatility procedures. Algorithms such as the one
developed here, are particularly attractive, because they can model
high dimensional data, with low computational cost, and still they
can enjoy the mathematical properties of closed estimation
procedures, which aim to address volatility estimation and
forecasting for a wide class of financial data.

\renewcommand{\theequation}{A-\arabic{equation}} % redefine the command that creates the equation no.
\setcounter{equation}{0}  % reset counter
\section*{Appendix}

\begin{proof}[Proof of Theorem \ref{lem1}]
Consider the transformation $Y=X^{1/2}S^{-1}X^{1/2}$. From Olkin and
Rubin (1964) the determinant of the Jacobian matrix of $X$ with
respect to $Y$ is $J(Y\rightarrow X)=J(Y\rightarrow
X^{1/2})J(X^{1/2}\rightarrow X)=\prod_{i\leq
j}^p(\lambda_i+\lambda_j)(\xi_i+\xi_j)^{-1}$, where
$\lambda_1,\ldots,\lambda_p$ are the eigenvalues of
$S^{-1/2}X^{1/2}S^{-1/2}$ and $\xi_1,\ldots,\xi_p$ are the
eigenvalues of $X^{1/2}$. We observe that if $A=I_p$, then $p(X)$ is
an inverted Wishart distribution, since
$\textrm{tr}(-X^{-1/2}SX^{-1/2}/2)=\textrm{tr}(-SX^{-1}/2)$. The
Jacobian $J(Y\rightarrow X)$ does not depend on $A$ and so we can
determine $J(Y\rightarrow X)$ from the special case of $A=I_p$. With
$A=I_p$, $X\sim \IW_p(n,S)$ and $Y\sim \IW_p(n,I_p)$ and from the
transformation $Y=X^{1/2}S^{-1}X^{1/2}$ we get
$$
p(Y)=\frac{|S|^{(n-p-1)/2}\textrm{etr}(-Y^{-1}/2)J(Y\rightarrow
X)}{2^{p(n-p-1)/2}\Gamma_p\{(n-p-1)/2\}|S|^{n/2}}.
$$
Since $Y\sim \IW_p(n,I_p)$ it must be
$|S|^{-n/2}|S|^{(n-p-1)/2}J(Y\rightarrow X)=1$ and so
$J(Y\rightarrow X)=|S|^{(p+1)/2}$.

Now, in the general case of a covariance matrix $A$, we see
$$
\int_{X>0} p(X)\,dX=\int_{Y>0}
\frac{|A|^{(n-p-1)/2}}{2^{p(n-p-1)/2}\Gamma_p\{(n-p-1)/2\}|Y|^{n/2}}
\textrm{etr}(-AY^{-1}/2)\,dY=1,
$$
since $Y\sim \IW_p(n,A)$.
\end{proof}

\begin{proof}[Proof of Theorem \ref{lem2a}]
First we prove (a). From the proof of Theorem \ref{lem1} we have that\\
$Y=X^{1/2}S^{-1}X^{1/2}\sim \IW_p(n,A)$ and so
$\E(Y)=(n-2p-2)^{-1}A$ and $\E(Y^{-1})=(n-p-1)A^{-1}$. Proceeding
with (b) we note from the proof of Theorem \ref{lem1} that for any
$n>2p$
$$
\int_{X>0}|X|^{-n/2}\textrm{etr}(-AX^{-1/2}SX^{-1/2}/2)\,dX=c^{-1},
$$
where $c$ is the normalizing constant of the distribution of $X$.
Then
$$
\E|X|^{\ell}=c\int_{X>0}|X|^{-(n-2\ell)/2}\textrm{etr}(-AX^{-1/2}SX^{-1/2}/2)\,dX=\frac{c}{c^*},
$$
where
$$
c^*=\frac{2^{p\ell}|A|^{(n-p-1)/2}|S|^{(n-p-1)/2}}{2^{p(n-p-1)/2}|A|^{\ell}|S|^{\ell}\Gamma_p\{(n-2\ell-p-1)/2\}}
$$
and the range of $\ell$ makes sure that $n-2\ell>2p$. The result
follows by eliminating the factor $2^{p(n-p-1)/2}$ in the fraction
$c/c^*$.
\end{proof}

\begin{proof}[Proof of Theorem \ref{th:commute}]
Suppose that $X\sim \GIW_p(n,A,S)$. From the normalizing constant of
the density $f(X)$ of Theorem \ref{lem1}, we can exchange the roles
of $|A|$ and $|S|$. And from $\textrm{tr}(-AX^{-1/2}S$ $\times
X^{-1/2}/2) =\textrm{tr}(-SX^{-1/2}AX^{-1/2}/2)$ we have that $X\sim
\GIW_p(n,S,A)$.
\end{proof}

In order to prove Theorem \ref{th:uhlig} we prove the somewhat more
general result in the following lemma.

\begin{lem}\label{lemma:uhlig}
Let $A_1\sim\W_p(m,I_p)$, $A_2=\sum_{j=1}^nY_jY_j'$, with
$Y_j\sim\N_p(0,I_p)$ and $H\sim\GW_p(m+n,A,S)$, where $A_1$, $Y_j$
$(j=1,\ldots,n)$ and $H$ are independent. Define $C=A_1+A_2$,
$B=\{\U(C)'\}^{-1}A_1\{\U(C)\}^{-1}$, $G=\U(H)'B\U(H)$ and
$D=H^{1/2}AH^{1/2}-G^{1/2}AG^{1/2}$. Then $C\sim\W_p(m+n,I_p)$,
$G\sim\N_p(0,S)$, where $C,G$ and $Z_j$ $(j=1,\ldots,n)$, are
independent.
\end{lem}

\begin{proof}
The proof mimics the proof of Uhlig (1994). Define
$Z_j=\U(H^{1/2}AH^{1/2})'\{\U(C)'\}^{-1}Y_j$ and note that
$D=\sum_{j=1}^nZ_jZ_j'$. From Theorem \ref{lem1} and from Uhlig
(1994), the Jacobian $J(A_1,H,Y_1,\ldots,Y_n\rightarrow
C,G,Z_1,\ldots,Z_n)$ is $|H|^{-n/2}|C|^{n/2}|A|^{-(p+1)/2}$. Then,
the joint density function of $A_1,H,A_2$ can be written as
\begin{gather*}
p(A_1,H,A_2) = p(A_1) p(H) p(A_2) = \left\{ 2^{pm/2}
\Gamma_p(m/2)\right\}^{-1} \textrm{etr}(-A_1/2) |A_1|^{(m-p-1)/2} \\
\times \left[ 2^{p(m+n)/2} \Gamma_p\{(m+n)/2\} |S|^{(m+n)/2}
\right]^{-1} |A|^{(m+n)/2} \textrm{etr}(-AH^{1/2}S^{-1}H^{1/2}/2)
|H|^{(m+n-p-1)/2} \\ \times (2\pi)^{-pn/2} \textrm{etr}(-A_2/2)
|A|^{-(p+1)} (\,dA_1) (\,dH) (\,dY_1) \cdots (\,dY_n) \\ = \left[
2^{p(m+n)/2} \Gamma_p\{(m+n)/2\}\right]^{-1}
\textrm{etr}(-C/2)|C|^{(m+n-p-1)/2} \\ \times \left\{ 2^{pm/2}
\Gamma_p(m/2)|S|^{m/2} \right\}^{-1} |A|^{m/2}
\textrm{etr}(-AG^{1/2}S^{-1}G^{1/2}/2) |G|^{(m-p-1)/2} \\ \times
(2\pi)^{-pn/2} |S|^{-n/2} \textrm{etr}(-S^{-1}D/2) |A|^{(n-p-1)/2} =
p(C)p(G)p(D),
\end{gather*}
where $A_1=|C||B|$, $H^{1/2}AH^{1/2}=G^{1/2}AG^{1/2}+D$ and
$|H|=|G|/|B|$ are used.
\end{proof}

\begin{proof}[Proof of Theorem \ref{th:uhlig}]
The proof is immediate from Lemma \ref{lemma:uhlig}, after noticing
that with the definition of the multivariate singular beta
distribution (Uhlig, 1994), $B\sim\B_p(m/2,n/2)$.
\end{proof}

Let $A>0$ denote that the matrix $A$ is positive definite and let
$A>B$ denote that the matrices $A>0$ and $B>0$ satisfy $A-B>0$. The
following two lemmas are needed in order to prove the limit of
Theorem \ref{lem:limit}.

\begin{lem}\label{lem2}
If the $p\times p$ matrices $A,B>0$ satisfy $A> B$, then $A^{-1}<
B^{-1}$.
\end{lem}
The proof of this lemma is given in Horn and Johnson (1999).

\begin{lem}\label{lem3}
If $P_{t}=R_{t}(R_{t}+I_p)^{-1}$, with $R_{t}=\phi^2
P_{t-1}+\Omega$, where $\Omega$ is a positive definite matrix and
$\phi$ is a real number, then the sequence of $p\times p$ positive
matrices $\{P_t\}$ is convergent.
\end{lem}

\begin{proof}
First suppose that $\phi=0$. Then $R_t=\Omega$, for all $t$, and so
$P_t=\Omega (\Omega+I_p)^{-1}$, which of course is convergent.

Suppose now that $\phi\neq 0$. It suffices to prove that $\{P_t\}$
is bounded and monotonic. Clearly, $0\leq P_{t}$ and since
$\phi^2>0$ and $\Omega$ is positive definite $0<P_{t}$, for all
$t>0$. Since $(R_{t}+I_p)^{-1}>0$,
$(R_{t}+I_p-R_{t})(R_{t}+I_p)^{-1}>0\Rightarrow
P_{t}=R_{t}(R_{t}+I_p)^{-1}<I_p$ and so $0<P_{t}<I_p$. For the
monotonicity it suffices to prove that, if $P_{t-1}^{-1}>
P_{t-2}^{-1}$ (equivalent $P_{t-1}^{-1}< P_{t-2}^{-1}$), then
$P_{t}^{-1}> P_{t-1}^{-1}$ (equivalent $P_{t}^{-1}< P_{t-1}^{-1}$).
From $P_{t-1}^{-1}> P_{t-2}^{-1}$ we have $P_{t-1}<
P_{t-2}\Rightarrow R_{t}< R_{t-1}\Rightarrow R_{t}^{-1}>
R_{t-1}^{-1}\Rightarrow
P_{t}^{-1}-P_{t-1}^{-1}=R_{t}^{-1}-R_{t-1}^{-1}> 0$, since
$P_{t}^{-1}=(R_{t}+I_p)R_{t}^{-1}=I_p+R_{t}^{-1}$. With an analogous
argument we have that, if $P_{t-1}^{-1}< P_{t-2}^{-1}$, then
$P_{t}^{-1}-P_{t-1}^{-1}< 0$, from which the monotonicity follows.
\end{proof}

\begin{lem}\label{lem3post}
Let $\{P_t\}$ be the sequence of Lemma \ref{lem3} and suppose that
$P_0=p_0I_p$, for a known constant $p_0>0$. Then, with $\Omega$ as
in Lemma \ref{lem3}, the limiting matrix
$P=\lim_{t\rightarrow\infty}P_t$ commutes with $\Omega$.
\end{lem}

\begin{proof}
First we prove that if $P_{t-1}$ commutes with $\Omega$, then
$P_{t}$ also commutes with $\Omega$. Indeed from $P_{t}=(\phi^2
P_{t-1}+\Omega)(\phi^2 P_{t-1}+\Omega+I_p)^{-1}$ we have that
$P_{t}^{-1}=I_p+(\phi^2 P_{t-1}+\Omega)^{-1}$ and then
$$
P_{t}^{-1}\Omega^{-1}=\Omega^{-1}+(\phi^2\Omega
P_{t-1}+\Omega^2)^{-1}=\Omega^{-1}+(\phi^2
P_{t-1}\Omega+\Omega^2)^{-1}=\Omega^{-1}P_{t}^{-1}
$$
which implies that $\Omega
P_{t}=(P_{t}^{-1}\Omega^{-1})^{-1}=(\Omega^{-1}P_{t}^{-1})^{-1}=P_{t}\Omega$
and so $P_{t}$ and $\Omega$ commute. Because $P_0=p_0I_p$, $P_0$
commutes with $\Omega$ and so by induction it follows that the
sequence of matrices $\{P_t, t\geq 0\}$ commutes with $\Omega$.
Since $P=\lim_{t\rightarrow\infty}P_t$ exists (Lemma \ref{lem3}) we
have
$$
P\Omega=\lim_{t\rightarrow\infty}(P_t\Omega)=\lim_{t\rightarrow\infty}(\Omega
P_t)=\Omega P
$$
and so $P$ commutes with $\Omega$.
\end{proof}

\begin{proof}[Proof of Theorem \ref{lem:limit}]
From Lemma \ref{lem3} we have that $P$ exists and from Lemma
\ref{lem3post} we have that $P$ and $\Omega$ commute. From
$P_{t}=(\phi^2P_{t-1}+\Omega)(P_{t-1}+\Omega+I_p)^{-1}$ we have
$P=(\phi^2P+\Omega)(\phi^2P+\Omega+I_p)^{-1}$ from which we get the
equation $P^2+\phi^{-2}P(\Omega+I_p-\phi^2I_p)-\phi^{-2}\Omega=0$.
Now since $P$ and $\Omega$ commute we can write
\begin{gather*}
P^2+\phi^{-2}P(\Omega+I_p-\phi^2I_p)-\phi^{-2}\Omega=0 \Rightarrow
P^2+\frac{1}{2\phi^2}P(\Omega+(1-\phi^2)I_p) \\
+\frac{1}{2\phi^2}(\Omega+(1-\phi^2)I_p)
P+\frac{1}{4\phi^4}(\Omega+(1-\phi^2)I_p)^2-
\frac{1}{4\phi^4}(\Omega+(1-\phi^2)I_p)^2-\Omega=0\\
\Rightarrow \left(
P+\frac{1}{2\phi^2}(\Omega+(1-\phi^2)I_p)\right)^2=\frac{1}{4\phi^4}(\Omega+(1-\phi^2)I_p)^2+\Omega
\\ \Rightarrow P=\frac{1}{2\phi^2}\left[
\left\{(\Omega+(1-\phi^2)I_p)^2+4\Omega\right\}^{1/2}-\Omega-(1-\phi^2)I_p\right],
\end{gather*}
after rejecting the negative definite root.
\end{proof}

\begin{proof}[Proof of Theorem \ref{th3}]
The proof is inductive in the distribution of $\Sigma_{t}|y^{t}$.
Assume that given $y^{t-1}$ the distribution of $\Sigma_{t-1}$ is
$\Sigma_{t-1}|y^{t-1}\sim \GIW((1-\delta)^{-1}+2p,Q^{-1},S_{t-1})$
and so
$\Sigma_{t-1}^{-1}|y^{t-1}\sim\GW_p((1-\delta)^{-1}+p-1,Q,S_{t-1}^{-1})$.
From the evolution (\ref{evol}) and Theorem \ref{th:uhlig}, we have
$\Sigma_{t}^{-1}|y^{t-1}\sim\GW_p(\delta(1-\delta)^{-1},Q,kS_{t-1}^{-1})$,
which proves that
$\Sigma_{t}|y^{t-1}\sim\GIW_p(\delta(1-\delta)^{-1},Q^{-1},k^{-1}S_{t-1})$.

From the Kalman filter, conditionally on $\Sigma_{t}$, the one-step
forecast density of $y_{t}$ is
$$
y_{t}|\Sigma_{t},y^{t-1}\sim
\N_p(m_{t-1},\Sigma_{t}^{1/2}Q_{t-1}(1)\Sigma_{t}^{1/2})\approx
\N_p(m_{t-1},\Sigma_{t}^{1/2}Q\Sigma_{t}^{1/2}),
$$
where $m_{t-1}$, $Q_{t-1}(1)$ and $Q$ are as in the theorem.

Given $y^{t-1}$ the joint distribution of $y_{t}$ and $\Sigma_t$ is
\begin{eqnarray}
p(y_{t},\Sigma_{t}|y^{t-1})&=&p(y_{t}|\Sigma_{t},y^{t-1})p(\Sigma_{t}|y^{t-1})\nonumber\\
&=&c_1\frac{\textrm{etr}\{-Q^{-1}\Sigma_{t}^{-1/2}
(e_{t}e_{t}'+k^{-1}S_{t-1})\Sigma_{t}^{-1/2}/2\}}{|\Sigma_{t}|^{(n+1)/2}},\label{eq9}
\end{eqnarray}
where $n=\delta (1-\delta)^{-1}+2p$ and
$$
c_1=\frac{|k^{-1}S_{t-1}|^{(n-p-1)/2}}{(2\pi)^{\pi/2}2^{p(n-p-1)}|Q|^{(n-p)/2}
\Gamma_p\{(n-p-1)/2\}}.
$$
The one-step forecast density of $y_{t}$ is
\begin{eqnarray*}
p(y_{t}|y^{t-1})&=&\int_{\Sigma_{t}>0} p(y_{t},\Sigma_{t}|y^{t-1})\,d\Sigma_{t}\\
&=& c_1 \int_{\Sigma_{t}>0} |\Sigma_{t}|^{-(n+1)/2}
\textrm{etr}\{-Q^{-1}\Sigma_{t}^{-1/2}(e_{t-1}e_{t-1}'+k^{-1}\Sigma_{t})\Sigma_{t}^{-1/2}/2\}
\,d\Sigma_{t} \\ &=& c_1 \frac{2^{p(n-p)/2} \Gamma_p\{(n-p)/2\} }{
|Q|^{-(n-p)/2} |e_{t-1}e_{t-1}'+k^{-1}S_{t-1}|^{(n-p)/2} } \\ &=&
\frac{\Gamma_p\{(n-2p+p)/2\} } { \pi^{p/2} \Gamma_p\{(n-2p+p-1)/2\}
} |k^{-1}S_{t-1}|^{(n-2p+p-1)/2}
|e_{t-1}e_{t-1}'+k^{-1}S_{t-1}|^{-(n-2p+p)/2},
\end{eqnarray*}
and so $y_{t}|y^{t-1}\sim\T_p (\delta (1-\delta)^{-1}, m_{t-1},
k^{-1}S_{t-1})$, as required. This completes (a).

Proceeding with (b) first we derive the distribution of
$\Sigma_{t}|y^{t}$. Applying the Bayes' theorem we have
$$
p(\Sigma_{t}|y^{t})=\frac{p(y_{t}|\Sigma_{t},y^{t-1})p(\Sigma_{t}|y^{t-1})}{p(y_{t}|y^{t-1})}
$$
and from equation (\ref{eq9}) we have
$$
p(\Sigma_{t}|y^{t})=c_2|\Sigma_{t}|^{-n^*/2}\textrm{etr}(-Q^{-1}\Sigma_{t}^{-1/2}S_{t}
\Sigma_{t}^{-1/2}/2)
$$
and
$$
n^*=n+1=\frac{\delta}{1-\delta}+2p+1=\frac{1}{1-\delta}+2p,
$$
where $S_{t}$ is as in the theorem and the proportionality constant
is $c_2=c_1/p(y_{t}|y_{t-1})$, not depending on $\Sigma_{t}$. Thus
$\Sigma_{t}|y^{t}\sim \GIW_p((1-\delta)^{-1}+2p,Q^{-1},S_{t})$ as
required. Conditionally on $\Sigma_{t}$, the distribution of
$\theta_{t}$ follows directly from application of the Kalman filter
and so applying the approximation $\Sigma_{t}\approx S_{t}^*$, with
$S_t^*$ as in the theorem, provides the required posterior
distribution of $\theta_{t}$.
\end{proof}

Before we prove Theorem \ref{th2}, we give the following lemma.

\begin{lem}\label{lemma}
Suppose that the $p\times p$ matrix $B$ follows the singular
multivariate beta distribution $B\sim\B_p(m/2,n/2)$, with density
$$
p(B)= \pi^{(n^2-pn)/2}
\frac{\Gamma_p\{(m+n)/2\}}{\Gamma_n(n/2)\Gamma_p(m/2)}
|K|^{(n-p-1)/2} |B|^{(m-p-1)/2},
$$
where $n$ is a positive integer, $m>p-1$, $I_p-B=H_1KH_1'$, $K$ is
the diagonal matrix with diagonal elements the positive eigenvalues
of $I_p-B$, and $H_1$ is a matrix with orthogonal columns, i.e.
$H_1H_1'=I_p$. For any non-singular matrix $A$, the density of
$X=AB^{-1}A'$, is
$$
p(X)=\pi^{(n^2-pn)/2}
\frac{\Gamma_p\{(m+n)/2\}}{\Gamma_n(n/2)\Gamma_p(m/2)} |A|^{n+m-p-1}
|L|^ {-(p-n+1)/2} |X|^{-(m-p-1)/2},
$$
where $L$ is the diagonal matrix including the positive eigenvalues
of $I_p-A'X^{-1}A$.
\end{lem}
\begin{proof}
First note that $X$ is a non-singular matrix and
$|B|=|A|^2|X|^{-1}$. From D\'{i}az-Garc\'{i}a and Guti\'{e}rrez
(1997), the Jacobian of $B$ with respect to $X$ is
$$
(\,dB)=|K|^{(p-n+1)/2} |L|^{-(p-n+1)/2} |A|^n (\,dX),
$$
where $K$ is defined as in the theorem. Then from the singular
multivariate beta density of $B$ we obtain
\begin{eqnarray*}
p(X)&=&\pi^{(n^2-pn)/2}
\frac{\Gamma_p\{(m+n)/2\}}{\Gamma_n(n/2)\Gamma_p(m/2)} |A|^n
|K|^{(n-p-1)/2} |B|^{(m-p-1)/2} \\ && \times |K|^{(p-n+1)/2}
|L|^{-(p-n+1)/2},
\end{eqnarray*}
from which we immediately get the required density of $X$.
\end{proof}

\begin{proof}[Proof of Theorem \ref{th2}]
The likelihood function is
\begin{equation}\label{logl1}
L(\Sigma_1,\ldots,\Sigma_N;y^N)=p(y_1|\Sigma_1)p(\Sigma_1|\Sigma_0)\prod_{t=2}^N
p(y_t|\Sigma_t,y^{t-1})p(\Sigma_t|\Sigma_{t-1},y^{t-1})
\end{equation}
and from the Kalman filter we have
$y_t|\Sigma_t,y^{t-1}\sim\N_p(m_{t-1},\Sigma_t^{1/2}Q\Sigma_t^{1/2})$,
where $Q_{t-1}(1)\approx Q$. The density
$p(\Sigma_t|\Sigma_{t-1},y^{t-1})$ is the density $p(X)$ of Lemma
\ref{lemma} with $A=k^{-1/2}\{\U(\Sigma_{t-1}^{-1})^{-1}\}'$,
$\Sigma_t^{-1}=\U(\Sigma_t^{-1})'\U(\Sigma_t^{-1})$,
$m=\delta(1-\delta)^{-1}+p-1$ and $n=1$. The required formula of the
log-likelihood function is obtained from (\ref{logl1}) by taking the
logarithm of $L(\Sigma_1,\ldots,\Sigma_N;y^N)$.
\end{proof}


\begin{thebibliography}{}

\bibitem{Aguilar00}
Aguilar, O. and West, M. (2000) Bayesian dynamic factor models and
portfolio allocation. {\it Journal of Business and Economic
Statistics}, {\bf 18}, 338-357.

\bibitem{Anderson79}
Anderson, B.D.O. and Moore, J.B. (1979) {\it Optimal Filtering.}
Prentice Hall, Englewood Cliffs NJ.

\bibitem{Asai06}
Asai, M,, McAleer, M. and Yu, J. (2006) Multivariate stochastic
volatility: A review. {\it Econometric Reviews}, {\bf 25}, 145-175.

\bibitem{Bauwens06}
Bauwens, L., Laurent, S. and Rombouts, J.V.K. (2006) Multivariate
GARCH models: A survey. {\it  Journal of Applied Econometrics}, {\bf
21}, 79-109.

\bibitem{Brown1994}
Brown, P.J., Le, N.D. and Zidek, J.V. (1994) Inference for a
covariance matrix, in: P.R. Freeman and A.F.M. Smith, eds. {\it
Aspects of Uncertainty}. Wiley, Chichester.

\bibitem{Carvalho1}
Carvalho, C.M. and West, M. (2007) Dynamic matrix-variate graphical
models. {\it Bayesian Analysis}, {\bf 2}, 69-98.

\bibitem{Chib07}
Chib, S., Omori, Y. and Asai, M. (2007) Multivariate stochastic
volatility. CIRJE Discussion Paper F-488 (permanent website: {\tt
http://www.e.u-tokyo.ac.jp/cirje/research/03research02dp.html}).

\bibitem{Dawid93}
Dawid, A.P. and Lauritzen, S.L. (1993) Hyper Markov laws in the
statistical analysis of decomposable graphical models. {\it Annals
of Statistics}, {\bf 21}, 1272-1317.

\bibitem{Diaz97}
D\'{i}az-Garc\'{i}a, J.A. and Guti\'{e}rrez, J.R. (1997) Proof of
the conjectures of H. Uhlig on the singular multivariate beta and
the jacobian of a certain matrix transformation. {\it Annals of
Statistics}, {\bf 25}, 2018-2023.

\bibitem{Durbin01}
Durbin, J. and Koopman, S.J. (2001) {\it Time Series Analysis by
State Space Methods}. Oxford University Press, Oxford.

\bibitem{Fildes92}
Fildes, R. (1992) The evaluation of extrapolative forecasting
methods. {\it International Journal of Forecasting}, {\bf 8}, 69-80.

\bibitem{Franses}
Franses, P.H. and van Dijk, D. (2000) {\it Nonlinear Time Series
Models in Empirical Finance.} Cambridge University Press, Cambridge.

\bibitem{Gupta}
Gupta, A.K. and Nagar, D.K. (1999) {\it Matrix Variate
Distributions}. Chapman and Hall, New York.

\bibitem{Harvey89}
Harvey, A.C. (1989) {\it Forecasting Structural Time Series Models
and the Kalman Filter}. Cambridge University Press, Cambridge.

\bibitem{Horn}
Horn, R.A. and Johnson, C.R. (1999) {\it Matrix Analysis.} Cambridge
University Press, Cambridge.

\bibitem{Letac}
Letac, G. and Massam, H. (2004) All invariant moments of the Wishart
distribution. {\it Scandinavian Journal of Statistics}, {\bf 31},
295-318.

\bibitem{Lutkepohl}
L\"utkepohl, H. (2007) {\it New Introduction to Multiple Time Series
Analysis}. Springer-Verlag, New-York.

\bibitem{Maasoumi06}
Maasoumi, E. and McAleer, M. (2006) Multivariate stochastic
volatility: An overview. {\it Econometric Reviews}, {\bf 25},
139-144.

\bibitem{Olkin64}
Olkin, I. and Rubin, H. (1964) Multivariate beta distributions and
independence properties of Wishart distribution. {\it Annals of
Mathematical Statistics}, {\bf{35}}, 261-269.

\bibitem{Philipov}
Philipov, A. and Glickman, M.E. (2006) Multivariate stochastic
volatility via Wishart processes. {\it Journal of Business and
Economic Statistics}, {\bf 24}, 313-328.

\bibitem{Roverato}
Roverato, A. (2002) Hyper inverse Wishart distribution for
non-decomposable graphs and its application to Bayesian inference
for Gaussian graphical models. {\it Scandinavian Journal of
Statistics}, {\bf 29}, 391–411.

\bibitem{Salvador1}
Salvador, M. and Gargallo, P. (2004). Automatic monitoring and
intervention in multivariate dynamic linear models. {\it
Computational Statistics and Data Analysis}, {\bf 47}, 401-431.

\bibitem{Srivastava03}
Srivastava, M.S. (2003) Singular Wishart and multivariate beta
distributions. {\it Annals of Statistics}, {\bf 31}, 1537-1560.

\bibitem{triantafyllopoulos06}
Triantafyllopoulos, K. (2006) Multivariate control charts based on
Bayesian state space models. {\it Quality and Reliability
Engineering International}, {\bf 22}, 693-707.

\bibitem{Triantafyllopoulos07}
Triantafyllopoulos, K. (2007) Feedback quality adjustment with
Bayesian state-space models. {\it Applied Stochastic Models in
Business and Industry}, {\bf 23}, 145-156.

\bibitem{Uhlig94}
Uhlig, H. (1994) On singular Wishart and singular multivariate beta
distributions. {\it Annals of Statistics}, {\bf{22}}, 395-405.

\bibitem{Uhlig97}
Uhlig, H. (1997) Bayesian vector autoregressions with stochastic
volatility. {\it Econometrica}, {\bf{65}}, 59-73.


\end{thebibliography}
\end{document}